\documentclass[a4paper,11pt]{article}

\usepackage[sc]{mathpazo}
\usepackage[T1]{fontenc}
\usepackage{amsthm,amsmath,amssymb,amsfonts,hyperref,a4wide}
\usepackage[usenames,dvipsnames]{xcolor} 
\usepackage{tikz}

\usetikzlibrary{matrix,arrows,backgrounds,positioning,fit,shapes}

\newtheorem{theorem}{Theorem}[section]
\newtheorem{lemma}[theorem]{Lemma}

\theoremstyle{definition}

\theoremstyle{remark}

\newtheorem{obs}{Observation}[section]

\newcommand*{\ind}{\text{ind}}

\title{Approximate counting using Taylor's theorem: \\ a survey}
\author{
Viresh Patel\footnote{School of Mathematical Sciences, Queen Mary, University of London. Email: \texttt{viresh.patel@qmul.ac.uk}. } \and Guus Regts\footnote{Korteweg de Vries Institute for Mathematics, University of Amsterdam. Email: \texttt{guusregts@gmail.com}. Funded by the Netherlands Organisation of Scientific Research (NWO): VI.Vidi.193.068}}

\begin{document}

\maketitle
\begin{abstract}
In this article we consider certain well-known polynomials associated with graphs including the independence polynomial and the chromatic polynomial. These polynomials count certain objects in graphs: independent sets in the case of the independence polynomial and proper colourings in the case of the chromatic polynomial. They also have interpretations as partition functions in statistical physics. 

The algorithmic problem of (approximately) computing these types of polynomials has been studied for close to 50 years, especially using Markov chain techniques. Around eight years ago, Barvinok devised a new algorithmic approach based on Taylor's theorem for computing the permanent of certain matrices, and the approach has been applied to
various graph polynomials since then. 
This article is intended as a gentle introduction to the approach as well as a partial survey of associated techniques and results.

\noindent \begin{footnotesize}
\end{footnotesize}
\end{abstract}

\section{Introduction}
Computational counting is an area of theoretical computer science, which, at its heart, is concerned with the computational problem of counting certain structures inside some combinatorial object given as input. Think of counting the number of satisfying assignments of some logical formula, or the number of independent sets in a graph.
Often the structures to be counted have some natural weighting and one is interested in the weighted count. 

In this article, we focus on graph counting problems and in particular on finding efficient algorithms for (approximately) counting  objects of interest  inside some input graph.
The counting problems we consider here are ones where the number of objects to be counted is typically super-polynomial in the size of the graph and cannot be directly enumerated in polynomial time.
For example, the number of independent sets of a graph can be exponentially large in the size of the graph\footnote{As a contrasting example the number of triangles of a graph is polynomial in the size of the graph and can be enumerated by brute force in polynomial time. Of course it is interesting to know whether there is an algorithm for counting triangles that is better than using brute force, but we do not pursue this here.} and indeed, the problem of (approximately) counting independent sets of a graph is a rich area of research (here an independent set in a graph is a subset of vertices no two of which are adjacent). 
 The problem of exact counting is often computationally hard (this is the case for independent sets~\cite{Roth,Vadhan,DyerGreenhill}) so one is usually interested in approximation algorithms for counting problems. A notable exception is the problem of counting spanning trees of a graph. The number of spanning trees is typically exponential in the size of the graph, but spanning trees can be counted in polynomial time via the matrix tree theorem~\cite{JBook}. Throughout the article we use the example of counting independent sets in graphs to illustrate the various ideas we discuss. 
  In fact the ideas apply more generally for counting many other graph theoretic objects including trees, matchings, cuts, and proper colourings (we discuss proper colourings towards the end of the article).


The basic combinatorial counting problems are often not treated directly, but are considered in more generality by examining their corresponding generating functions. For example, for independent sets, one is interested in the independence polynomial, which for a graph $G=(V,E)$, is defined to be the polynomial 
\[
Z_G(\lambda) 
:= \sum_{S \subseteq V \text{independent}} \lambda^{|S|}
= \sum_{k \geq 0} \alpha_k \lambda^k,
\]
where $\alpha_k = \alpha_k(G)$ is the number of independent sets of size $k$ in $G$. This polynomial encodes a lot of information about the (sizes) of independent sets in $G$. For example it is easy to see that $Z_G(1)$ gives the number of independent sets in $G$ and $Z_G'(1)/Z_G(1)$ gives the average size of an independent set in $G$. Knowing the value of $Z_G(\lambda)$ for very large $\lambda$ would allow one to extract the degree of the polynomial i.e.\ the size of the largest independent set (which is known to be $NP$-hard to compute and even to approximate within a constant factor). This already tells us we should not expect to be able to efficiently approximate the independence polynomial at all values of $\lambda$. 

In this article, we describe a recent technique, the so-called Taylor polynomial interpolation method of Barvinok (first introduced in \cite{Barper}), for designing approximation algorithms for computational counting problems. Our aim here is to introduce the reader to the ideas behind the method and to give a flavour of the mathematics involved. We do not intend to give a complete survey of results that use the technique and nor do we fully formalise all of the ideas we present.
For the latter, we refer the reader to the excellent book of Barvinok~\cite{Barbook} and to~\cite{PatReg17}.

One distinguishing feature of the Taylor interpolation method is that, as well as its applications to ordinary counting problems, it also applies to evaluations of generating functions at negative and complex numbers. This is in contrast to earlier techniques. One motivation for understanding such complex evaluations is in quantum computing~\cite{Welshbook,BFLWquantum,mann2019approximation,GGHPotts}, although we will not discuss this here. Another is that complex evaluations are sometimes useful for real counting problems (see e.g.~\cite{realholant}), and perhaps most importantly, broadening our perspective to the complex plane gives a deeper understanding of the underlying computational complexity of various counting problems (see Section~\ref{se:conc} for more discussion on this).


Other techniques for designing approximate counting algorithms (which we will not discuss) include the Markov chain Monte Carlo method (see Jerrum~\cite{JBook} for an excellent introduction to the area) as well as the correlation decay method first introduced by Weitz~\cite{Weitz} and Bandyopadhyay and Gamarnik~\cite{BanGar} (see e.g.\ Chapters 5 and 6 of~\cite{Barbook} for an introduction). 
A very recent technique, closely related to Barvinok's interpolation method, is based on the cluster expansion from statistical physics and has been introduced by Jenssen, Keevash and Perkins~\cite{JKP20}. We say a few words about this at the end of Section~\ref{se:zerofree}.

\subsection{Connection to statistical physics}
The generating functions for the counting problems we encounter are often studied in the statistical physics community (using different terminology). For example the independence polynomial is known as the partition function of the hard-core model in statistical physics. The hard-core model is a model for gases. Given a closed container of a gas at equilibrium consider examining the gas in a small region of space inside the container. The (discretised) space in the region is represented by a grid graph, where vertices of the graph represent points in space. Each such point can either be occupied or unoccupied by a gas molecule but adjacent points in space cannot both be occupied due to repulsive forces between the molecules. Therefore, at any moment in time, the gas molecules can only occupy an independent set in the grid. The probability $\mathbb{P}(S)$ that at any moment in time the occupied points form a particular independent set $S$ of the grid is proportional to $\lambda^{|S|}$, where $\lambda \in [0, \infty)$ is a temperature-like parameter often called the \emph{fugacity}. A high temperature corresponds to a small value of $\lambda$, which, as we intuitively expect, makes it less likely that we see a large set $S$ of occupied points in our small region of space. Since $\mathbb{P}(S) \propto \lambda^{|S|}$, and $\sum_{S\subseteq V \text{ independent}} \mathbb{P}(S) = 1$, we see that $\mathbb{P}(S) = \lambda^{|S|} / Z_G(\lambda)$. Here we see the independence polynomial $Z_G(\lambda)$ (known here as the partition function of the hard-core model) appearing as the normalising constant in the probability. Again, this partition function is much more than just a normalising constant, and encodes a lot of physical information about the system. For example, by considering the limiting behaviour of $\ln Z_G(\lambda)/|V(G)|$ and its derivatives for larger and larger graphs (usually grids), discontinuities of these limit functions give information about phase transitions in the system, that is, sharp changes in the physical parameters associated with the system indicating a qualitative change in the system.
We direct the reader to \cite{FVbook} for a comprehensive and rigourous mathematical treatment of phase transitions for many models and to \cite{ScottSokal,Dobrushin96,GauntFisher}  for more on the hard-core model.
We will not be concerned directly with the statistical physics, but some results originally proved by statistical physicists will be used in the algorithmic approach we describe.

\subsection{Preliminaries} 
\label{se:prelim}
We have already mentioned the independence polynomial as an example of a graph polynomial that we may wish to approximate. The independence polynomial will serve as a running example throughout the article to illustrate various ideas. Here we mention a few basic properties of the independence polynomial to give the reader a feel for this object.

Recall that $Z_G(\lambda) = \sum \lambda^{|S|}$, where the sum is over all independent sets $S$ of $G$. The first easy but important fact to note is that the empty set is an independent set, so $Z_G(0)$ (i.e.\ the constant term in the polynomial) is always $1$. Another important fact is that the independence polynomial is multiplicative, that is $Z_{G_1 \cup G_2}(\lambda) = Z_{G_1}(\lambda)Z_{G_2}(\lambda)$, where we write $G_1 \cup G_2$ for the disjoint union of the graphs $G_1$ and $G_2$. This is because every independent set $S$ of $G_1 \cup G_2$ can be written uniquely as $S = S_1 \cup S_2$, where $S_i$ is an independent set of $G_i$. Therefore $\lambda^{|S|} = \lambda^{|S_1|}\lambda^{|S_2|}$, which allows us to factorise the sum. Using this multiplicative property, we also see, for example, that the independence polynomial of $k$ isolated vertices is $(1+ \lambda)^k$. One can also see directly that the complete graph on $k$ vertices has independence polynomial $1 + k\lambda$.


We now describe the type of algorithm we ideally wish to obtain for our graph counting problems. Suppose $p = p(G)$ is a graph parameter, e.g.\ $p(G)$ is the number of independent sets in $G$, or $p(G) = Z_G(\lambda)$ for some fixed $\lambda$. Note that we allow $p(G)$ to be a complex number. A \emph{fully polynomial-time approximation scheme} (or FPTAS for short) for $p$ is an algorithm that takes as input a graph $G$ and an error tolerance $\varepsilon >0$ and outputs a (complex) number $N$ such that $N = e^{\varepsilon t}p(G)$ for some $t \in \mathbb{C}$ with $|t| \leq 1$ in time polynomial in $|G|$ (the number of vertices of $G$) and $\varepsilon^{-1}$. 
Note that when $\varepsilon$ is small, we have $N = e^{\varepsilon t}p(G) \approx (1 + \varepsilon t)p(G)$, so that $N$ is roughly within a distance $\varepsilon |p(G)|$  of the true value of $p(G)$. For this reason we call such output $N$ a multiplicative $\varepsilon$-approximation (for $p(G)$).
\footnote{Note that this definition of FPTAS is consistent with the usual notion of FPTAS for real parameters. 
} 
We also discuss algorithms that provide the same approximation as above but that run in time super-polynomial in $|G|$.


\section{Barvinok's interpolation method}
\label{se:Barv}

In this section we describe the Taylor polynomial interpolation method of Barvinok, a method that can be applied to a wide variety of counting problems. Consider some graph polynomial, that is, each graph $G$ has some associated polynomial $P(z)= P_G(z)$. As with the independence polynomial, we should imagine that $P_G$ is not directly accessible, i.e.\ at least some of its coefficients are difficult to compute from $G$. We will however assume that the degree of the polynomial $P_G$ is always bounded by a constant times $|G|$; this is certainly the case for the independence polynomial and is easy to verify for most graph polynomials one might consider. Our goal is to (efficiently) obtain a multiplicative $\varepsilon$-approximation for $P_G(z)$ for $z \in \mathbb{C}$. 

The insight of Barvinok was to use Taylor's theorem, about power series approximations of smooth functions, to obtain the desired approximation. At first sight we seem to gain nothing from Taylor's theorem because the Taylor series of a polynomial is simply the polynomial itself. However, notice that the truncated Taylor series of a (non-polynomial) function gives an \emph{additive} $\varepsilon$-approximation to the function, whereas we are interested in a \emph{multiplicative} $\varepsilon$-approximation. Therefore, rather than considering the Taylor series of $P_G(z)$, we should in fact consider the Taylor series of $g(z) := \ln P_G(z)$ 
and then take the exponential of the result to obtain the desired approximation.\footnote{In order for $g(z)$ to be well-defined we need to fix a branch of the logarithm here; we say more below.}

To this end, consider the Taylor series of $g(z)$ about zero:
\[
g(z) = \sum_{k = 0}^{\infty} \frac{g^{(k)}(0)}{k!} z^k,
\]
where $g^{(k)}$ denotes the $k$th derivative of $g$. Unfortunately, the Taylor series of a function does not usually converge for all $z \in \mathbb{C}$. We will return shortly to the question of convergence, but let us assume for now that the Taylor series does converge to $g(z)$ for a value of $z$ we are interested in. In that case, if we write $T_m(z)$ for the first $m$ terms of the Taylor series of $g$ above, then for $m$ sufficiently large, we will have that $|T_m(z) - g(z)| < \varepsilon$, i.e.\ $T_m(z) = g(z) + \varepsilon t$ for some $t \in \mathbb{C}$ with $|t|<1$. Taking the exponential of both sides of the equation, we obtain $\exp(T_m(z)) =  \exp(\varepsilon t) P_G(z)$ i.e.\ $\exp(T_m(z))$ a multiplicative $\varepsilon$-approximation for $P_G(z)$.

This gives us the desired approximation, but several questions remain. Firstly, there is the question of convergence mentioned above. Secondly, if the Taylor series does converge, then how large does $m$ have to be to guarantee that $|T_m(z) - g(z)| < \varepsilon$? Finally, how can we actually compute $T_m(z)$ in order to compute our approximation $\exp(T_m(z))$ for $P_G(z)$?  Note that we do not have direct access to the numbers $g^{(k)}(0)$; these have to be computed in some way.

For the first question of convergence, Taylor's theorem says that the Taylor series for $g$ converges inside the disk $D_R := \{z \in \mathbb{C}: |z| \leq R \}$ for any $R > 0$ provided that $g$ is analytic inside $D_R$. In our case, this holds provided $P_G(z) \not= 0$ for all $z \in D_R$.\footnote{Formally, to ensure $g$ is analytic, we fix $\ln P_G(0)$, and take the branch of $g(z) =\ln P_G(z)$ on $D_R$ given by $g(z) = \ln P_G(0)+\int^z_{0}P_G'(w)/P_G(w) dw$.}
So the Taylor series will converge inside the largest disk that contains no roots of $P_G(z)$. Establishing such zero-freeness results for particular graph polynomials will be the subject of Section~\ref{se:zerofree}.

The second question concerns the rate of convergence of the Taylor series of $g$. Here we take advantage of the particular form of $g$ as the logarithm of a polynomial. If $\eta_1, \ldots, \eta_d$ are the (complex) roots\footnote{Again, we do not typically have access to the roots of $P_G$; we work with the roots only in the analysis of the algorithm.} of $P_G(z)$ then we can write $P_G(z) = a\prod_{i=1}^d(1 - \frac{z}{\eta_i})$, where $a = P_G(0)$ is assumed to be non-zero. Then taking logarithms of both sides, we have
\[
g(z) = \ln(a) + \sum_{i=1}^d \ln(1 - (z/\eta_i)).
\]
Using that the Taylor series of $\ln(1-z) = -z - \frac{z^2}{2} - \frac{z^3}{3} - \cdots$ for $|z|<1$, we obtain the Taylor series of $g$ as
\[
g(z) = \ln(a) - \sum_{i=1}^d \sum_{k = 1}^{\infty} \frac{(z/\eta_i)^k}{k}
\]
for $|z| < \min_i |\eta_i|$ (precisely the condition of zero-freeness mentioned above). Assuming $|z| \leq \delta \min_i |\eta_i|$ for some $\delta \in (0,1)$, this gives 
\[
|g(z) - T_m(z)| 
\leq \sum_{i=1}^d \sum_{k=m}^{\infty} \left| \frac{(z/\eta_i)^k}{k} \right| 
\leq \sum_{i=1}^d \sum_{k=m}^{\infty} \delta^k = \frac{d \delta^m}{1 - \delta}.
\]
In order to bound the last expression by $\varepsilon$, it is sufficient to take $m \geq C \ln(d / \varepsilon)$ for some constant $C$ depending on $\delta$. For such $m$ we have that $\exp(T_m(z))$ is a multiplicative $\varepsilon$-approximation for $P_G(z)$.

The final question of actually computing $T_m(z)$ is more subtle and will only be partially addressed here and in the next section. We will show that if we know the values of the first $m = C\ln(d/\varepsilon)$ coefficients of $P_G(z)$ then we can compute the derivatives $g^{(0)}, g^{(1)}, \ldots, g^{(m)}$ in time ${\rm poly}(m)$. However, we do not typically have immediate access to the first $m$ coefficients of our graph polynomials. For example, in the case of the independence polynomial $Z_G(\lambda)$, the coefficient $\alpha_k$ of $\lambda^k$ is the number of independent sets of size $k$ in $G$: computing this naively with a brute force approach of checking every $k$-tuple of vertices takes time $n^k$ (where $n = |G|$) and so computing $\alpha_m$ takes time $n^{m} = n^{O(\ln n)}$ (noting that the degree of $Z_G$ i.e.\ the size of the largest independent set could be and often is linear in $n$). In the next section, we show how to compute $\alpha_0, \ldots, \alpha_m$ in ${\rm poly}(n)$ time and the idea turns out to generalise for many other graph polynomials of interest. For now, here is how to compute $T_m(z)$ given the first $m$ coefficients of $P_G(z)$.

Suppose $P(z) = P_G(z) = a_0 + a_1z + \cdots + a_dz^d$. We defined $g(z) = \ln P_G(z)$. We know $g^{(0)}(0) = g(0) = \ln(a_0)$. If we differentiate once and rearrange, we obtain $g^{(1)}(z)P(z) = P^{(1)}(z)$. If we now repeatedly differentiate this expression, we obtain the following expressions:
\begin{align*}
P^{(1)} &= g^{(1)}P^{(0)} \\
P^{(2)} &= g^{(2)}P^{(0)} + g^{(1)}P^{(1)} \\
& \vdots \\
P^{(r)} &= g^{(r)}P^{(0)} + \binom{r-1}{1} g^{(r-1)}P^{(1)} + \binom{r-1}{2} g^{(r-2)}P^{(2)} + \cdots + \binom{r-1}{r-1}g^{(1)}P^{(r-1)}.
\end{align*}    
Evaluating these expressions at zero, and noting that $P^{(r)}(0) = r!a_r$, we obtain
\begin{align*}
a_1 &= a_0g^{(1)}(0) \\
2a_2 &= a_0g^{(2)}(0) + a_1g^{(1)}(0) \\
& \vdots \\
ra_r &= a_0g^{(r)}(0) + \frac{(r-1)!}{(r-1)!}a_1 g^{(r-1)}(0) + \frac{(r-1)!}{(r-2)!}a_2 g^{(r-2)}(0) + \cdots + \frac{(r-1)!}{1!}a_{r-1}g^{(1)}(0).
\end{align*}    
We see that if we know $a_0, \ldots, a_r$ and we have computed $g^{(0)}(0), \ldots, g^{(r-1)}(0)$, then we can use the $r$th equation above to compute $g^{(r)}(0)$ in time $O(r)$. Therefore given $a_0, \ldots, a_m$, we can compute $T_m(z)$ in $O(m^2)$ time.

The following summarises what we have shown in this section and is the essence of the Taylor polynomial interpolation method.

\begin{theorem}
\label{thm:Barvinok}
Suppose $\mathcal{G}$ is an (infinite) set of graphs and for each $G \in \mathcal{G}$, $P_G(z)$ is a polynomial associated with $G$, where 
\[
P_G(z) = \sum_{i=0}^{d(G)} a_i(G)z^i
\]
 Suppose there exists $R>0$ and a function $T: \mathbb{N} \times \mathbb{N} \rightarrow \mathbb{N}$ with the properties that
\begin{itemize}
    \item[(i)] $P_G(z) \not= 0$ whenever $|z| \leq R$ for all graphs $G \in \mathcal{G}$, and
    \item[(ii)] we are able to compute $a_i(G)$ in time bounded by $T(|G|, i)$, where we assume for convenience that $T$ is non-decreasing in both arguments.
\end{itemize}
Then there is an algorithm, which, given input $G \in \mathcal{G}$, $\varepsilon >0$, and $z \in \mathbb{C}$ with $|z| < R$, computes a multiplicative $\varepsilon$-approximation of $P_G(z)$ in time 
$mT( n, m ) + O(m^2)$, where $n = |G|$ and $m := C \ln(d(G) / \varepsilon)$ (as defined earlier).
\end{theorem}

Some remarks are in order. The theorem is formulated for a general class of graphs $\mathcal{G}$ rather than all graphs because often, we are only able to establish conditions (i) and (ii) effectively for certain types of graphs (typically bounded degree graphs). This is best illustrated by applying the result above to the independence polynomial.

For the independence polynomial, if we consider $\mathcal{G}$ to be the set of all graphs, then there is no zero-free disk of positive radius\footnote{The $k$-vertex complete graph has independence polynomial $1 + k\lambda$, so its roots tend to $0$.} so we can only take $R=0$ in condition (i). However, if we restrict $\mathcal{G}$ to graphs of maximum degree $\Delta$, we will see in Section~\ref{se:zerofree} that we can take $R = (\Delta - 1)^{\Delta - 1} / \Delta^{\Delta}$. Similarly for condition (ii), if we take $\mathcal{G}$ to be all graphs, then the brute force approach mentioned earlier is essentially the only way to compute coefficients of $Z_G(\lambda)$ giving $T(n,k) = O(n^k)$: overall this gives a super-polynomial running time of $mT(n,m) + O(m^2) = n^{O(m)} + O(m^2) = n^{O(\ln(n/\varepsilon))}$.
Such a quasi-polynomial running time is already quite promising because it is significantly better than the exponential running time of a brute-force algorithm.
However, in the next section we will see that for graphs of maximum degree at most $\Delta$, we can compute the coefficients much faster and take $T(n,k) = {\rm poly}(n) \Delta^{O(k)}$, thereby establishing an overall polynomial running time of $T(n,m) = {\rm poly}(n) \Delta^{O(m)} = {\rm poly}(n)\Delta^{O(\ln(n / \varepsilon))} = (n / \varepsilon)^{O(\ln \Delta)} $. Combining the results from the next two sections will therefore give an FPTAS for computing $Z_G(\lambda)$ on graphs of maximum degree at most $\Delta$ provided $|\lambda| < (\Delta - 1)^{\Delta - 1} / \Delta^{\Delta}$.

Finally, we remark that one can in fact relax condition (i) to include  regions that are not necessarily disks provided the region is ``thick'' in a certain sense and contains the point $0$. 
Concretely one should think of a small neighbourhood of a real interval or a sector region. Relaxing condition (i) to non-disk regions is achieved by making suitable polynomial transformations of $P_G$; see Section 2.2.2 of \cite{Barbook} and~\cite{Barrealrooted} for details.




\section{Polynomial running time for bounded degree graphs}
\label{se:polytime}

In the last section we saw how we can use Taylor's theorem to design algorithms to approximate graph polynomials. Let $P = P_G$ be a graph polynomial. Examining Theorem~\ref{thm:Barvinok}, we require two ingredients to establish an approximation algorithm to compute $P_G(z)$. First we need to establish a zero-free disk for $P_G$; this will be discussed in detail in the next section. Second, we need to be able to efficiently compute the first $O(\ln |G|)$ coefficients of $P_G$, which we discuss in detail here. There is usually a straightforward, direct approach for computing these coefficients, which leads to quasi-polynomial time algorithms, but which is not fast enough for an FPTAS. 
We already saw this in the last section with the independence polynomial, where we saw that computing the coefficients naively leads to an $n^{O(\ln n)}$-time algorithm.\footnote{Computing the coefficients naively often gives us quasi-polynomial time approximation algorithms as with the example of the independence polynomial. This is already very good because it is a significant improvement on the exponential time taken to enumerate independent sets in a graph. Achieving the polynomial runtime of an FPTAS is considered to be the gold standard in the area.} In this section, we show how to compute the coefficients of the independence polynomial more efficiently for {\it bounded degree graphs}. The technique generalises to many other graph polynomials but all the key ideas are best understood through the concrete example of the independence polynomial. We give the statement for general graph polynomials at the end of the section. 

It is worth noting that, barring a few exceptions, the setting of bounded degree graphs is often the setting of interest. For example, the problem of computing a multiplicative $\varepsilon$-approximation for $Z_G(z)$ is known to be computationally hard for all complex $z \not= 0$ if we have no restriction on $G$; see~\cite{BeyondlambdacSlyandSun,BeyondlambdacGalanisetal,BGGS20}.

\subsection{Computing the coefficients of the independence polynomial efficiently}

Recall that the independence polynomial $Z_G(\lambda)$ is given by 
\[
Z_G(\lambda) = \sum_{k \geq 0} \alpha_k \lambda^k,
\]
where $\alpha_k = \alpha_k(G)$ is the number of independent sets of size $k$ in $G$. Throughout this section we focus on bounded degree graphs and write $\mathcal{G}_{\Delta}$ for the set of graphs of maximum degree at most $\Delta$. If we apply Theorem~\ref{thm:Barvinok} to $Z_G(\lambda)$ (with $G \in \mathcal{G}_{\Delta}$), assuming we have some suitable zero-free disk containing $\lambda$, Theorem~\ref{thm:Barvinok} gives us an algorithm to compute a multiplicative $\varepsilon$-approximation of $Z_G(\lambda)$ in time $mT(n,m) + O(m^2)$, where $T(n,i)$ is the time needed to compute $\alpha_i(G)$ for $n$-vertex graphs $G \in \mathcal{G}_{\Delta}$ and $m \leq C\ln(n / \varepsilon)$. We will sketch a proof of the following.

\begin{theorem}
\label{thm:fastind}
For $G \in \mathcal{G}_{\Delta}$, we can compute $\alpha_i(G)$ in time ${\rm poly}(|G|) \Delta^{O(i)}$, i.e.\ for $n$-vertex graphs of maximum degree at most $\Delta$, we can take $T(n,i) = {\rm poly}(n) \Delta^{O(i)}$.
\end{theorem}
Using the theorem above, Theorem~\ref{thm:Barvinok} gives us an approximation algorithm for the independence polynomial with running time 
\[
mT(n,m) + O(m^2) = {\rm poly(n)} \Delta^{O( \ln(n / \varepsilon)} = {\rm poly(n)}(n/ \varepsilon)^{O(\ln( \Delta))}.
\]
We see that this running time is of the form required for a fully polynomial time approximation scheme.
We now sketch the proof of Theorem~\ref{thm:fastind}.

\subsubsection*{Sketch proof of Theorem~\ref{thm:fastind}}
\label{se:sketchproof}
We begin by generalising the algorithmic problem we are interested in. We are interested in computing $\alpha_k(G)$, the number of independent sets of size $k$ in $G$, when $G \in \mathcal{G}_{\Delta}$. Equivalently, $\alpha_k(G)$ is the number of induced copies of the graph $I_k$ in $G$, where $I_k$ is the graph consisting of $k$ vertices and no edges. Generally for graphs $H$ and $G$, write $\ind(H,G)$ for the number of induced copies\footnote{The number of induced copies of a graph $H$ in the graph $G=(V,E)$ is defined as the number of vertex subsets $S \subseteq V$ such that $G[S] = H$.} of $H$ in $G$. Then $\alpha_k(G) = \ind(I_k, G)$.

The first observation is that, while we do not know how to efficiently compute $\ind(I_k, G)$, it is not too hard to efficiently compute $\ind(H,G)$, when $H$ is connected.

\begin{obs}
\label{obs:1}
We can compute $\ind(H,G)$ in time ${\rm poly}(n) \Delta^{O(k)}$, where $G \in \mathcal{G}_{\Delta}$, $n = |G|$ and $k = |H|$.
\end{obs}

To see this, we first pick any spanning tree $T$ in $H$ (i.e.\ a subgraph of $H$ that is a tree and that uses all $k$ vertices of $H$). Such a spanning tree exists because $H$ is connected. The idea is to find all (not necessarily induced) copies of $T$ in $G$ and to check which of the copies of $T$ extend to an induced copy of $H$. This accounts for all induced copies of $H$ because every induced copy of $H$ in $G$ contains a (not necessarily induced) copy of $T$ in $G$.

There are only relatively few (not necessarily induced) copies of $T$ in $G$. Indeed, first we enumerate the vertices of $T$ in a breadth-first ordering $v_1, v_2, \ldots, v_k$. We embed $T$ into $G$ one vertex at a time in order. There are $n$ choices of where to embed $v_1$. Each subsequent vertex of $T$ has at most $\Delta$ possibilities for its embedding into $G$ because when we come to embed $v_i$, its parent in $T$ (say $v_{i'}$) has already been embedded as some vertex $x_{i'}$ in $G$, so the embedding of $v_i$ must be a neighbour of $x_{i'}$ in $G$. Therefore altogether there are at most $n \Delta^{k-1}$ embeddings of $T$ in $G$ and each such embedding of $T$ is checked to see if it gives an induced copy of of $H$. 

From Observation~\ref{obs:1}, we see that we can compute $\ind(H,G)$, when $H$ is connected, but the graph $H$ we are interested in, namely $I_k$, is very much disconnected. It would be useful if we could express $\ind(I_k, G)$ in terms of $\ind(H,G)$ for connected $H$. A trivial case of this is the fact that $\ind(I_2,G) = \binom{n}{2} - \ind(e,G)$, where $e$ is the graph on two vertices with an edge between them. This says nothing other than that the number of edges and non-edges in an $n$-vertex graph sum to $\binom{n}{2}$. With a little more work, we can express $\ind(I_3, G)$ in terms of induced counts of connected graphs as follows. There are four graphs on three vertices, namely $I_3$, the triangle denoted $T$, the path on three vertices denoted $P_3$ and the disjoint union of an edge and a vertex denoted $e + I_1$. By enumerating all induced subgraphs of $G$ on three vertices, we have
\[
\ind(I_3, G) = \binom{n}{3} - \ind(T,G) - \ind(P_2, G) - \ind(e + I_1, G). 
\]
The only disconnected graph on the right hand side is $e + I_1$, and by simple counting, it is not too hard to show that
\[
\ind(e+I_1, G) = (n-2)\ind(e,G) - 2 \ind(P_3, G) - 3 \ind(T,G).
\]
Substituting the second formula into the first gives an expression for $\ind(I_3, G)$ in terms of induced counts of connected graphs.

These calculations suggest that it is possible to express $\ind(I_k,G)$ in terms of induced counts $\ind(H,G)$ for connected graphs $H$, but that the calculations and formulae will get cumbersome. 
A new idea is needed to approach the problem in a systematic and manageable way. The next observation is the key insight to overcoming this hurdle and is at the heart of the proof of Theorem~\ref{thm:fastind}. It was proved by Csikv{\'a}ri and Frenkel~\cite{CsikFren}; the proof is short and can also be found in \cite{PatReg17}.

\begin{obs}
\label{obs:2}
Suppose $\tau(G)$ is an additive graph property, meaning that it satisfies the following two properties.
\begin{itemize}
    \item[(i)] $\tau(G)$ can be written as sum of products of induced graph counts, i.e.\ for all $G$ 
    \[
    \tau(G) = \sum_{i=1}^r \mu_i \prod_{H \in \mathcal{H}_i}\ind(H,G),
    \]
    where $\mathcal{H}_i$ is a (finite) set of graphs and $\mu_i \in \mathbb{C}$ is a constant for each $i = 1, \ldots, r$, and
    \item[(ii)] $\tau(G_1 \cup G_2) = \tau(G_1) + \tau(G_2)$ for all graphs $G_1$ and $G_2$. 
\end{itemize}
Then $\tau$ is in fact of a simpler form, namely, for all graphs $G$, we have
\[
\tau(G) = \lambda_1 \ind(H_1, G) + \cdots + \lambda_s \ind(H_s, G),
\]
where $H_1, \ldots, H_s$ are {\it connected} graphs and $\lambda_1, \ldots, \lambda_s \in \mathbb{C}$.
\end{obs}
The observation above says that every additive graph parameter is a linear combination of $\ind(H_i, G)$ for {\it connected} $H_i$, and so by Observation~\ref{obs:1}, such additive graph parameters can be computed efficiently.\footnote{Actually efficient computation is not immediate because it depends on the number and size of the $H_i$; we address this later.}
Our task now is reduced to the task of expressing $\alpha_k(G) = \ind(I_k, G)$ in terms of additive graph parameters. In order to do this, we now switch from the combinatorial to the polynomial perspective of $\alpha_k(G)$. 

Recall that the $\alpha_k(G)$ are the coefficients of the independence polynomial $Z_G$, i.e.\ $Z_G(\lambda) = \alpha_0 + \alpha_1 \lambda + \cdots \alpha_d \lambda^d$. Suppose that $\eta_1, \ldots, \eta_d$ are the roots of $Z_G$. Noting that the constant term $\alpha_0$ is one, we can write $Z_G(\lambda) = (1 - \eta_1^{-1}\lambda) \cdots (1 - \eta_d^{-1}\lambda)$. While we cannot compute the $\eta_i$ directly, we can relate them to the coefficients $\alpha_k$ by expanding the product above. We see that the $\alpha_k$ are the elementary symmetric polynomials in $\eta_i^{-1}$, namely 
\[
\alpha_0 = 1, 
\hspace{1 cm}  \alpha_1 = -\sum_{1 \leq i \leq d} \eta_i^{-1},
\hspace{1 cm} \alpha_2 = \sum_{1 \leq i < j \leq d} \eta_i^{-1}\eta_j^{-1}
\hspace{1 cm} \text{etc}.
\]
Another important class of symmetric polynomials are the power sums. Let us define the $i$th power sum $p_i$ to be 
\[
p_i = \eta_1^{-i} + \cdots + \eta_d^{-i}.
\]
It is well known that the power sums can be related to the elementary symmetric polynomials using the Newton identities. There are several short derivations of these identities. 
In the context of our problem, the Newton identities give the following expressions relating the $\alpha_i$ and the $p_i$.
\begin{align*}
-\alpha_1 &= \alpha_0 p_1 \\
-2\alpha_2 &= \alpha_0 p_2 + \alpha_1 p_1 \\
-3\alpha_3 &= \alpha_0 p_3 + \alpha_1 p_2 + \alpha_2 p_1 \\
& \vdots \\
-t\alpha_t &= \alpha_0 p_t + \alpha_1 p_{t-1} + \cdots + \alpha_{t-1} p_1.
\end{align*}
From this it is easy to see that if we know the values of the $p_i$ then we can inductively compute the $\alpha_i$. Indeed, if we know the values of $p_1, \ldots, p_t$, and we also know (by induction) the values of $\alpha_1 \ldots, \alpha_{t-1}$ then using the $t^{th}$ identity, we can compute $\alpha_t$. Thus the problem of efficiently computing the $\alpha_i$ is reduced to that of efficiently computing the $p_i$. It is possible to efficiently compute the power sums because, as the reader may have guessed, the power sums are additive graph parameters.

\begin{obs}
\label{obs:3}
The power sums $p_i = p_i(G)$ as defined above have the property of being  additive graph parameters.
\end{obs}
It is easy to verify that $p_i$ satisfies the second property of an additive graph parameter, namely that $p_i(G_1 \cup G_2) = p_i(G_1) + p_i(G_2)$ for any graphs $G_1$ and $G_2$. Indeed, since $Z_{G_1 \cup G_2} = Z_{G_1}Z_{G_2}$ (see Section~\ref{se:prelim}), if $\eta_1, \ldots, \eta_d$ are the roots of of $Z_{G_1}$ and $\nu_1, \ldots, \nu_{d'}$ are the roots of $Z_{G_2}$ then $\eta_1, \ldots, \eta_d, \nu_1, \ldots, \nu_{d'}$ are the roots of $Z_{G_1 \cup G_2}$ so that
\[
p_i(G_1 \cup G_2) = \eta_1^{-i} + \cdots + \eta_d^{-i} + \nu_1^{-i} + \cdots + \nu_{d'}^{-i} = p_i(G_1) + p_i(G_2).
\]
For the first property, we use the Newton identities. Note that, since $\alpha_0=1$, we can rearrange the $t^{\rm th}$ identity and express $p_t$ as a sum of products of $p_1, \ldots, p_{t-1}$ and $\alpha_1, \ldots, \alpha_t$. We know that the $\alpha_i$ are induced graph counts, and if we assume by induction that $p_1, \ldots, p_{t-1}$ are also sums of products of induced graph counts, then we see that $p_t$ is also a sum of products of induced graph counts and so satisfies property (i) of an additive graph parameter.

We now have all the ingredients to explain how to compute the $\alpha_k$ efficiently. We can compute the power sums $p_i$ efficiently. This is because the power sums are additive graph parameters (Observation~\ref{obs:3}) and they are therefore linear combinations of induced counts of {\it connected} graphs (Observation~\ref{obs:2}). Each induced graph count $\ind(H,G)$ in this linear combination can be computed efficiently when $G$ is of bounded degree since $H$ is connected (Observation~\ref{obs:1}) thus allowing us to compute the power sums efficiently. Once we have computed the power sums $p_1, p_2, \ldots$, we can inductively compute the $\alpha_i$ using the Newton identities. 

This gives the main ideas of the argument although there are a few subtleties that we have glossed over. The main one is that it is not quite obvious that we can compute the power sums $p_i(G)$ efficiently, i.e.\ in time ${\rm poly(|G|) } \Delta^{O(i)}$. While the $p_i(G)$ can be expressed as a linear combination of induced counts of connected graphs 
\[
p_i(G) = \lambda_1 \ind(H_1, G) + \cdots +\lambda_s \ind(H_s,G),
\]
we have not said how to find $H_1, \ldots, H_s$ and $\lambda_1, \ldots, \lambda_s$. 
Conceivably, $s$ could be superexponential in $i$ or the $H_i$ could have size superlinear in $i$; in either case we would not automatically get the desired running time. However, by using the Newton identities more carefully, and using the fact that $G$ has bounded degree it is not too difficult to overcome these technical obstacles. All the details can be found in \cite{PatReg17}.

\subsection{Computing the coefficients of other graph polynomials efficiently}

In Section~\ref{se:sketchproof}, we described the main idea of how we can efficiently compute the first $\ln |G|$ coefficients of the independence polynomial $Z_G$ for graphs $G$ of bounded degree. The ideas can be generalised to work for many other graph polynomials of interest.

What are the crucial properties of the independence polynomial $Z_G$ that we use in the sketch proof of Theorem~\ref{thm:fastind}? The whole proof is based around manipulating induced graph counts, so we certainly need the coefficients of $Z_G$ to be (functions of) induced graph counts. We also crucially need that $Z_G$ is multiplicative, which allows us to conclude that the power sums are additive, therefore allowing us to compute them efficiently.  

In \cite{PatReg17}, we show that if a graph polynomial $P = P_G$ satisfies certain  properties given below, then its coefficients can be computed efficiently for bounded degree graphs i.e.\ the $i$th coefficient of $P_G$ can be computed in time ${\rm poly}(n)\Delta^{O(i)}$ where $G$ is an $n$-vertex graph of maximum degree at most $\Delta$. As with the independence polynomial, this is enough to use the Taylor polynomial interpolation method from Section~\ref{se:Barv} to give an approximation algorithm for computing $P_G(z)$ (provided $z$ is in a suitable zero-free disk) with the required run time of an FPTAS. 

Suppose $P = P_G$ is a graph polynomial given by $P_G(z) = a_0 + a_1z + \cdots + a_dz^d$. Suppose that $P$ satisfies the following properties for some fixed constant $\alpha > 0$:
\begin{itemize}
    \item[(i)] for each $\ell$, the $\ell$th coefficient of $P$ can be expressed as a ``$\alpha$-bounded'' linear combination of induced graph counts, that is, for all $G \in \mathcal{G}_{\Delta}$
    \[
    a_{\ell}(G) = \sum_{H} \zeta_{H, \ell} \ind(H,G),
    \]
    where the sum is over graphs $H$ with at most $\alpha \ell$ vertices and
     $\zeta_{H, \ell} \in \mathbb{C}$ are constants (independent of $G$);
    \item[(ii)] in property (i), for each $H$ we can compute $\zeta_{H, \ell}$ in time $\exp(O(|H|)$; and
    \item[(iii)] $P_G$ is multiplicative, i.e.\ $P_{G_1 \cup G_2} = P_{G_1}P_{G_2}$.
\end{itemize}
Then we can compute $a_i(G)$ in time ${\rm poly}(|G|) \Delta^{O(i)}$. Again, using the Taylor polynomial interpolation method, this leads to an FPTAS for approximating $P_G(z)$ for $G \in \mathcal{G}_{\Delta}$, again provided we establish a suitable zero-free disk containing $z$. 

Note that in the case of the independence polynomial, properties (i) and (ii) are trivial and we saw it is easy to verify property (iii). These properties also hold for various other graph polynomials including the matching polynomial, the chromatic polynomial, and the Tutte polynomial.\footnote{The Tutte polynomial is a polynomial in two variables, but the properties above hold if one of the variables is fixed} We will not check these here, but refer the interested reader to \cite{PatReg17}. It is also worth noting that the technique described in this section can be adapted and applied to polynomials beyond those satisfying properties (i)-(iii) above; see \cite{LSSleeyang,BarReg,mann2019approximation}.

In Section~\ref{se:Barv} we explained how one can design algorithms for approximating graph polynomials using Taylor's theorem. In this section, we showed how to make these algorithms efficient (having the running time of an FPTAS) for many graph polynomials provided we restrict attention to bounded degree graphs. We have seen in Section~\ref{se:Barv} that essential to all of these algorithms is to establish a suitable zero-free disk or zero-free region in the complex plane for the graph polynomial in question. Our discussion of algorithms ends at this point and in the next section, we turn our attention entirely to the independent problem of establishing these zero-free regions.

\section{Techniques for proving absence of zeros}
\label{se:zerofree}
In the previous sections, we have sketched how the problem of approximately evaluating graph polynomials (particularly the independence polynomial) in a region of the complex plane is reduced to the problem of establishing that the polynomial has no zeros in that region.
There is a long history of proving such results about the locations of zeros of graph polynomials and partition functions. The techniques used often have their origin in statistical physics but have now been picked up and extended by the theoretical computer science community. In this section we will discuss three different techniques.

\subsection{Recursion and ratios}
Many graph polynomials satisfy recursions in which the polynomial for a given graph can be expressed in terms of the polynomial for smaller graphs. Such recursions allow us  to prove properties about the graph polynomial, such as absence of zeros, by induction. However, rather than working with the polynomials directly, it is often more productive to work instead with related quantities.
We illustrate this approach through our running example of the independence polynomial and at the end of the section we direct the reader to further work in which this technique is employed.

Our aim is to sketch a proof of the following result due to Shearer~\cite{Shearer}, Dobrushin~\cite{Dobrushin96} and Scott and Sokal~\cite{ScottSokal}:
\begin{theorem}\label{thm:shearer}
Let $G=(V,E)$ be a graph with maximum degree $\Delta \geq 2$ and  let $\lambda\in \mathbb{C}$ satisfy $|\lambda|\leq \lambda^*(\Delta):=\frac{(\Delta-1)^{\Delta-1}}{\Delta^{\Delta}}$. Then $Z_G(\lambda)\neq 0$.
\end{theorem}

Let us briefly discuss this result before delving into the proof.
First, recall that by the Taylor polynomial interpolation method (particularly Theorem~\ref{thm:Barvinok} and Theorem~\ref{thm:fastind}), this result immediately implies an FPTAS for computing $Z_G(\lambda)$ for $G\in \mathcal{G}_\Delta$ inside the zero-free disk given by $|\lambda|< \lambda^*(\Delta)$. Second, note that if we are only interested in zero-free \emph{disks}, then one cannot improve Theorem~\ref{thm:shearer} in the sense that we cannot increase the constant $\lambda^*(\Delta)$.
Indeed, one can show that there is a sequence of graphs $G_n$ (in fact trees) of maximum degree $\Delta$ and negative numbers $\lambda_n$ such that $Z_{G_n}(\lambda_n)=0$ and $\lambda_n \to -\lambda^*(\Delta)$~\cite{ScottSokal}. However there has been a lot of interest recently in establishing zero-freeness for non-disk regions. Most notably, it was shown recently~\cite{PRSokal} that $Z_G(\lambda) \not= 0$ whenever $G \in \mathcal{G}_{\Delta}$ and $\lambda \in R \subseteq \mathbb{C}$ where $R$ is an open set containing the interval $[0, \lambda_c(\Delta))$ and $\lambda_c(\Delta) := (\Delta -1)^{\Delta-1 }/(\Delta- 2)^{\Delta}$. One significance of $\lambda_c(\Delta)$  is that it is an algorithmic threshold for real parameters $\lambda$: using the interpolation method, the result in \cite{PRSokal} implies  that there is an FPTAS\footnote{In fact, an FPTAS was established earlier in \cite{Weitz} using the correlation decay method.} to compute $Z_{G}(\lambda)$ whenever $G \in \mathcal{G}_{\Delta}$ and $\lambda \in [0, \lambda_c(\Delta))$, while for $\lambda > \lambda_c(\Delta)$, it is known that there is no such FPTAS unless $P = NP$ \cite{BeyondlambdacSlyandSun,BeyondlambdacGalanisetal}.

We now discuss the proof of Theorem~\ref{thm:shearer}.
Let $G=(V,E)$ be a graph and fix a vertex $v\in V$.
We can write down a recursion for $Z_G(\lambda) = \sum_{S \subseteq V \text{ independent }} \lambda^{|S|}$ by splitting the sum over those independent sets that do not contain $v$ and those that do to obtain
\begin{equation}\label{eq:fundamental id}
Z_G(\lambda)=Z_{G-v}(\lambda)+\lambda Z_{G\setminus [N[v]}(\lambda),
\end{equation}
where $G-v$ (resp. $G\setminus N[v]$) denote the graphs obtained from $G$ by removing $v$ (resp. $v$ and its neighbours in $G$). 
As mentioned earlier, rather than working directly with a recursion for $Z_G$, it turns out to be more useful to work with a recursion of a related quantity.
Define the \emph{ratio}, $R_{G,v}$, by
\begin{equation}\label{eq:ratio}
R_{G,v}(\lambda):=\frac{\lambda Z_{G\setminus N[v]}(  \lambda)}{Z_{G-v}(  \lambda)}.
\end{equation}
Observe that provided $Z_{G-v}(\lambda)\neq 0$, we have $Z_G(\lambda)=0$ if and only if $R_{G,v}(\lambda)=-1$ (using \eqref{eq:fundamental id}).
So to prove absence of zeros it suffices to inductively show that the ratios avoid $-1$.

Next we establish a recursion for these ratios.
Let $G$ be a graph with fixed vertex $u_0$ and let $\lambda\in \mathbb{C}$.
Let $u_1,\ldots,u_d$ be the neighbours of $u_0$ in $G$ (in any order). 
Set $G_0=G-u_0$ and define for $i=1,\ldots,d$, $G_i:=G_{i-1}-u_{i}$ (so $G_{d}=G\setminus N[u_0]$).
Suppose that $Z_{G_i}(\lambda)\neq 0$ for all $i=0,\ldots,d$.
Then we use `telescoping' to write 
\[
\frac{R_{G,u_0}(\lambda)}{\lambda}
= \frac{Z_{G_d}(\lambda)}{Z_{G_0}(\lambda)}
= \frac{Z_{G_1}(\lambda)}{Z_{G_0}(\lambda)}
\cdot \frac{Z_{G_2}(\lambda)}{Z_{G_1}(\lambda)}
\cdots \frac{Z_{G_{d}}(\lambda)}{Z_{G_{d-1}}(\lambda)}.
\]
Applying~\eqref{eq:fundamental id} to each of the denominators and after some rearranging we end up with the following identity:
\begin{equation}\label{eq:ratio recurse}
R_{G,u_0}(\lambda)=\frac{\lambda}{\prod_{i=1}^d (1+R_{G_{i-1},u_i}(\lambda))}.
\end{equation}

The identity above captures all the relevant combinatorics of independent sets that we need and the rest of the proof essentially boils down to proving a property about the above recursion. 


\begin{proof}[Proof of Theorem~\ref{thm:shearer}]
We may assume that $G$ is connected (if $G$ has connected components $H_1, \ldots, H_k$ then $Z_G(\lambda) = Z_{H_1}(\lambda) \cdots Z_{H_k}(\lambda)$ and so it is sufficient to prove the theorem for each $H_i$).

Fix $v_0\in V$.
We will show by induction that the following holds for all $U\subseteq V\setminus \{v_0\}$:
\begin{itemize}
\item[(i)] $Z_{G[U]}(\lambda)\neq 0$,
\item[(ii)]\text{if $u_0\in U$ has a neighbour in $V\setminus U$, then $|R_{G[U],u_0}(\lambda)|<1/\Delta$.}
\end{itemize}
Indeed if $|U|=0$ then this is trivially true, so
suppose that $|U|>0$.
Then since $G$ is connected, there is $u_0\in U$ that has a neighbour $v\in V\setminus U$.
Let us write $H=G[U]$ and let $u_1,\ldots,u_d$ be the neighbours of $u_0$ in $H$. Let $H_0=H-u_0$ and $H_i=H_{i-1}-u_i$ for $i>0$.
Then, by induction $Z_{H_i}(\lambda)\neq 0$ and $|R_{H_i,u_{i+1}}(\lambda)|<1/\Delta$ (since $u_{i+1}$ has a neighbour in $U \setminus V(H_i)$, namely $u_0$).
So we may use~\eqref{eq:ratio recurse} to conclude that
\begin{align}
|R_{H,u_0}(\lambda)|=\frac{|\lambda|}{\prod_{i=1}^d |1+R_{H_{i-1},u_i}(\lambda)|}&<|\lambda| (1-1/\Delta)^{-d}\nonumber
\\
&\leq |\lambda|\left( \frac{\Delta-1}{\Delta} \right)^{-(\Delta-1)} = 1/\Delta,\label{eq:derive}
\end{align}
where we used that $d \leq \Delta - 1$ (since $u_0$ has a neighbour in $V \setminus U$) and that $|\lambda| \leq \lambda_{\Delta}$. This shows (ii).
Then, we also see that $R_{H,u_0}(\lambda)\neq -1$ and so $Z_H(\lambda)\neq 0$, showing (i). This completes the induction.

To conclude the proof of the theorem we apply the same trick once more to $R_{G,v_0}$. From \eqref{eq:derive} we then obtain the bound $|R_{G,v_0}|<1/(\Delta-1)$ since $v_0$ may have $d = \Delta$ neighbours rather than $d \leq \Delta -1$.
Again we have $R_{G,v_0}(\lambda)\neq 1$ and so $Z_G(\lambda)\neq 0$, as desired.
\end{proof}

The proof essentially consists of two steps. First express a suitably chosen ratio in terms of ratios of smaller graphs. 
Secondly, use this expression to inductively show that these ratios are `trapped' in some suitable region of the complex plane (the open disk of radius $1/\Delta$ in the proof above).
Of course the real ingenuity comes in finding the right `trapping region'.

This approach can be traced back to work of Dobrushin~\cite{Dobrushin96} and possibly even earlier.
Recent years have seen many variations and refinements of this approach resulting in significant extensions of Theorem~\ref{thm:shearer}~\cite{PRSokal,bencs2018note,bencs2022complex} and zero-free regions for permanents~\cite{Barper,Barperhaf,Bardiagdom}, for the graph homomorphism partition functions~\cite{BarSobhom,BarSobhomwith}, for the partition function of the Ising and Potts models~\cite{LSSfisher,LSS2Delta,PRising,BarBarising,BDPR21,CDKPR}, for Holant problems~\cite{Regzero} and for various other graph polynomials~\cite{Barcliques,Barcube,BarReg,BarPella,Barmixed,li2021complex}.

\subsection{Stability of multivariate polynomials}

In this subsection we briefly mention the technique of polynomial stability without going into too much detail. The basic idea here is that there are certain operations on polynomials that preserve certain useful properties. If one can use these operations to construct some desired graph polynomial or partition function from ``elementary'' polynomials, we can establish useful properties of the graph polynomial / partition function.
The method is often most effective for multivariate polynomials, and indeed many graph polynomials have multivariate counterparts.

For our running example, the independence polynomial, the multivariate counterpart is defined as follows. Let $G=(V,E)$ be a graph and associate to each vertex $v$ a variable $x_v$.
The \emph{multivariate independence polynomial} is then defined as
\[
Z_G((x_v))=\sum_{\substack{S\subseteq V\\\text{independent}}} x^S,
\]
where we use the shorthand notation $x^S:=\prod_{v\in S} x_v$. Note that if we set all the variables equal to $\lambda$ then we recover the original (univariate) independence polynomial. The multivariate independence polynomial is a multi-affine polynomial meaning that it is affine in each variable (i.e.\ if we fix all but one variable $x_v$ it becomes a polynomial of degree $1$ in $x_v$). It is easy to see that any multi-affine polynomial $f$ (in the same variables $(x_v)_{v \in V}$) can be written as $f = \sum_{S\subseteq V} a_s x^S$ for some constants $a_S$.

For two multi-affine polynomials $P=\sum_{S\subseteq V} p_S x^S$ and $Q=\sum_{S\subseteq V}q_S x^S$, their \emph{Schur product}, $P*Q$ is defined as the multi-affine polynomial in which the coefficient of $x^S$ is  $p_S\cdot q_S$ i.e.\ $P*Q = \sum_{S\subseteq V} p_Sq_S x^S$.
We can build up the polynomial $Z_G$ using Schur product of simpler polynomials as follows. Suppose $H_1$ and $H_2$ are graphs on the same vertex set $V$ and $G$ is the union\footnote{This is very different from the disjoint union of graphs that we made heavy use of in Section~\ref{se:polytime}.} of $H_1$ and $H_2$ (i.e.\ the edges of $G$ are precisely the edges of $H_1$ together with the edges of $H_2$). Then 
\[
Z_G = Z_{H_1} * Z_{H_2}.
\]
This is easy to see since we know $S$ is an independent set of $G$ if and only if $S$ is an independent set of both $H_1$ and $H_2$ and the Schur product has the corresponding property that the coefficient of $x^S$ is $1$ in $Z_{H_1} * Z_{H_2}$ if and only if it is $1$ in both $Z_{H_1}$ and in $Z_{H_2}$.
For example the $4$-cycle $C_4$ with vertex set $\{1,2,3,4\}$ and edges $\{1,2\}$, $\{2,3\}$, $\{3,4\}$ and $\{4,1\}$ is the union of two matchings $M_1$ with edges $\{1,2\}, \{3,4\}$ and $M_2$ with edges $\{1,3\}, \{2,4\}$. Using the multiplicative property\footnote{We showed this property for the univariate independence polynomial and it follows in the same way for the multivariate version} of the independence polynomial, we know
\[
Z_{M_1}=(1+x_1+x_{2})(1+x_{3}+x_4) \quad \text{ and } \quad Z_{M_2}=(1+x_2+x_{3})(1+x_{1}+x_4)
\]
and using the Schur product property, one can check
\[
Z_{C_4} = Z_{M_1}*Z_{M_2} = (1 + x_1 + x_2 + x_3 + x_4 + x_1x_3 + x_2x_4). 
\]

The Schur product corresponds beautifully well to taking unions of graphs for the independence polynomial, but does it preserve any useful properties? Writing $\mathbb{D}$ for the open unit disk in $\mathbb{C}$, we say a multi-affine polynomial $P=\sum_{S\subseteq V} p_S x^S$ is $\mathbb{D}$-stable if $P((x_v)_{v\in V}) \not= 0$ whenever $x_v \in \mathbb{D}$ for all $v \in V$. It is well known (see~\cite{Barbook}) that if $P$ and $Q$ are $\mathbb{D}$-stable then so is $P*Q$. This seems promising for us, but unfortunately, the independence polynomial of a matching or indeed a single edge (out of which we build all other independence polynomials) is not $\mathbb{D}$-stable, e.g.\ $Z_{M_1}(-\frac{1}{2}, -\frac{1}{2}, 0, 0) = 0$. The independence polynomial of a matching is however non-zero if all the arguments are in an open disk of radius $1/2$. Now, using the fact that every graph in $\mathcal{G}_{\Delta}$ is the union of at most $\Delta + 1$ matchings (Vizing's theorem) and applying a simple scaling argument, one can still make use of the $\mathbb{D}$-stability of Schur products to show that $Z_G$ is non-zero if all arguments are in a disk of radius smaller than $1/2^{\Delta + 1}$, where $G \in \mathcal{G}_{\Delta}$.

This is a much weaker bound than  Theorem~\ref{thm:shearer} from the 
the previous subsection, but is given simply to illustrate the idea of stability.
The idea of using multi-affine polynomials and operations preserving zero-freeness was pioneered by Asano~\cite{Asano} about fifty years ago to give a short and elegant proof of the famous Lee-Yang theorem (see also~\cite{Barbook} for a proof using Schur products.) The theorem states that the partition function of the Ising model (in terms of vertex activities), which essentially is the generating function of the edge cuts in the graph, has all its zeros on the unit circle under suitable conditions; we choose not to introduce the relevant background here. 
By now there are several variations of the technique, some of which use the Grace-Sz\"ego-Walsh theorem, and they have been applied to partition functions of several models and graph polynomials~\cite{Ruellegraph-counting,Ruellegraph-counting,Wagner,ZerosHolant,ZerosHolant,BcR21}.

\subsection{The polymer method}

We introduced the multivariate independence polynomial in the last subsection to illustrate the idea of polynomial stability. It turns out that many other graph polynomials and partition functions can be expressed as evaluations of multivariate independence polynomials of a particular type. For this reason, there has been a lot of interest in understanding and proving conditions that guarantee zero-freeness of such multivariate independence polynomials.
This idea of first rewriting a partition function/graph polynomial as an evaluation of a multivariate independence polynomial and then checking conditions from the literature known to guarantee that the latter evaluation is nonzero is a powerful technique originating in statistical physics. There, the multivariate independence polynomial is sometimes called the partition function of a polymer model, and the technique we describe is sometimes called the polymer method.

We will give an example of this idea applied to the chromatic polynomial, a graph polynomial used for counting proper colourings of a graph, which we will shortly introduce.
We sketch a proof of a result of F\'ernandez and Procacci~\cite{FerProc} and Jackson, Procacci and Sokal~\cite{JPS13} about zero-freeness of the chromatic polynomial. At the end of the subsection, we list some recent results based on this technique and indicate how a variation of this technique can in fact be used directly to design efficient algorithms to approximate graph polynomials, without having to use the interpolation method.

\subsubsection{The chromatic polynomial}


For a graph $G=(V,E)$ and integer $q$, a proper $q$ colouring of $G$ is an assignment of $q$ colours (usually labelled $1, \ldots, q$) to the vertices such that adjacent vertices receive different colours. This means in particular that all vertices assigned some fixed colour $i$ form an independent set. The function $\chi_G$ counts the number of proper $q$-colourings of $G$, that is, for each $q \in \mathbb{N}$, $\chi_G(q)$ is defined to be the number of proper $q$-colourings of $G$. 
For example the number of proper $q$-colourings of a triangle is $q(q-1)(q-2)$ since after ordering the vertices arbitrarily, the first vertex can receive any of the $q$ colours, the second vertex may receive any of the colours except the colour of the first vertex, and the third vertex may receive any colour except those of the first two vertices (which are different).\footnote{Note that the formula is correct even when $q<3$, i.e.\ when there are no proper $q$-colourings of the triangle.} More generally, the number of proper $q$ colourings of $K_r$, the complete graph on $r$ vertices is $q(q-1) \cdots (q-r+1)$, i.e.\ $\chi_{K_r}(q) = q(q-1) \cdots (q-r+1)$ for every $q \in \mathbb{N}$. For any tree $T$ on $r$ vertices, $\chi_T(q) = q(q-1)^{r-1}$ for all $q \in \mathbb{N}$ since if we colour the vertices in a breadth-first ordering, then the first vertex may receive any of the $q$ colours, while each subsequent vertex can receive any colour except that of its parent. Of course, it is not usually so easy to determine $\chi_G(q)$ because it is NP-complete to decide if there is even one proper $q$-colouring of $G$, i.e.\ whether $\chi_G(q)$ is positive or not. Nonetheless, as the examples above suggest, $\chi_G(q)$ is always a polynomial in $q$ as we shall see shortly, and $\chi_G$ is called the chromatic polynomial of $G$. 

The chromatic polynomial was introduced in 1912 by Birkhoff in an attempt to prove the four colour theorem. It has a long history and has been studied from many perspectives together with its far-reaching generalisation, the Tutte polynomial (see \cite{chromaticbook,ellis2022handbook} for a comprehensive account). 

We now establish a very useful formula for the chromatic polynomial called the random cluster model, due to Fortuin and Kasteleyn~(see \cite{FortuinKasteleyn}); it is sometimes used as the definition of the chromatic polynomial. Formally, a proper $q$-colouring of a graph $G=(V,E)$ is a function $f: V \rightarrow \{1, \ldots, q\} =: [q] $ such that $f(u) \not= f(v)$ whenever $\{u,v\} \in E$. Then we can write 
\[
\chi_G(q) = \sum_{f: V \rightarrow [q]} \prod_{\{u,v\} \in E} \mathbb{1}_{f(u) \not= f(v)},
\]
where $\mathbb{1}_{f(u) \not= f(v)}$ is the indicator function that $f(u) \not = f(v)$ (so that the product is $1$ if and only if all edges are properly coloured). Replacing $\mathbb{1}_{f(u) \not= f(v)}$ with $(1 - \mathbb{1}_{f(u) = f(v))})$ and expanding, we obtain
\begin{align*}
\chi_G(q) = \sum_{f: V \rightarrow [q]} \prod_{\{u,v\} \in E} (1 - \mathbb{1}_{f(u) = f(v))})
&= \sum_{f: V \rightarrow [q]} \sum_{F \subseteq E} (-1)^{|F|}
\prod_{\{u,v\} \in F}  \mathbb{1}_{f(u) = f(v)} \\
&=  \sum_{F \subseteq E} (-1)^{|F|} \sum_{f: V \rightarrow [q]}
\prod_{\{u,v\} \in F}  \mathbb{1}_{f(u) = f(v)}.
\end{align*}
The inner sum in the last expression is equal to $q^{k(F)}$, where $k(F)$ is the number of components of the graph $(V,F)$. The reason is that the product is $1$ for an assignment $f$  if and only if every edge of $F$ is monochromatic in $f$, which means that $f$ must assign a single colour to each component of $(V,F)$. There are precisely $q^{k(F)}$ ways of doing this. Thus
\begin{equation}
\label{eq:RCM}
\chi_G(q):=\sum_{F\subseteq E} q^{k(F)}(-1)^{|F|},
\end{equation}
and so we see that $\chi_G$ is indeed a polynomial (although there are easier ways of showing this) and has degree $|G|$.  

Our goal will be to prove the following zero-freeness result for the chromatic polynomial.
\begin{theorem}[\cite{FerProc,JPS13}]\label{thm:JPS}
Let $G$ be any graph. Then all the zeros of $\chi_G$ are contained in the disk of radius $6.91\Delta(G)$ centered at $0$ in the complex plane.
\end{theorem}
It is likely that the constant $6.91$ can be improved, but it is not clear what the optimal value is likely to be; see \cite[Footnote 4]{RoyleSokal} for further discussion.
By the Taylor polynomial interpolation method, Theorem~\ref{thm:JPS} almost immediately implies an FPTAS for approximating $\chi_G(q)$ whenever $G \in \mathcal{G}_{\Delta}$ and $|q| \geq 6.92\Delta$. The trick is to apply the interpolation method to the polynomial $q^{|V|}\chi_G(1/q)$, which has no zeros in the disk of radius $\frac{1}{6.91\Delta}$.  From the  combinatorial perspective, this implies an FPTAS to count the number of proper $q$-colourings of any graph $G \in \mathcal{G}_{\Delta}$ whenever $q > 6.91 \Delta$. It is believed that there is an FPTAS for counting proper $q$-colourings whenever $q > \Delta$ and this is an active area of research.
By  proving a zero-freeness result for a different polynomial (the partition function of the Potts model) Liu, Sinclair, and Srivastava~\cite{liu2019correlation} have shown that there is an FPTAS when $q \geq 2\Delta$, and this is currently the state of the art.\footnote{There are improved bounds if we allow randomised algorithms based on the Markov chain Monte Carlo method~\cite{Vigoda,Improvedboundscoloring}.} 

We now sketch the proof of Theorem~\ref{thm:JPS}. As mentioned, we will need to work again with the (multivariate) independence polynomial and to make use of a suitable zero-freeness result for it.

\subsubsection{The chromatic polynomial as a multivariate independence polynomial}
Our first lemma shows how to express the chromatic polynomial of a graph $G$ as an evaluation of the multivariate independence polynomial of an associated graph. 
For this we need some notation. Let $G=(V,E)$ be a graph. Define a new graph $\Gamma$ whose vertices are subsets $S$ of $V$ of size at least two.
(In the context of the polymer method, these sets are called polymers.) Two of those sets $S,T$ are connected by an edge if and only if $S\cap T\neq \emptyset$. Notice that the graph $\Gamma$ is independent of the edges of $G$.

We now associate weights to vertices of $\Gamma$ as follows; these will depend on the edges of $G$ and on $q$ (the variable in the chromatic polynomial). For each vertex $S$ of $\Gamma$, i.e.\ $S \subseteq V$ with $|S| \geq 2$, define
\begin{equation}
\label{eq:weights}
\lambda_S:=\sum_{\substack{F\subseteq E(S)\\ \text{connected}}} (-1)^{|F|} q^{|S|-1}.
\end{equation}
Now the multivariate independence polynomial of $\Gamma$ with the (complex) vertex weights $\lambda_S$ is given by 
\begin{equation}\label{eq:def Z_G mult}
    Z_\Gamma((\lambda_S)) = \sum_{\substack{I \subseteq V(\Gamma) \\ \text{independent}}}\prod_{S\in I}\lambda_S.
\end{equation}

\begin{lemma}
\label{le:ChromIndep}
With notation as above we have
\[
q^{|V|}\chi_G(1/q)=Z_{\Gamma}((\lambda_S)).
\]
\end{lemma}
\begin{proof}
We start by expanding the left-hand side using \eqref{eq:RCM}
\begin{align*}
q^{|V|}\chi_G(1/q)&=\sum_{F\subseteq E}(-1)^{|F|} q^{|V|-k(F)}
=\sum_{F\subseteq E}\prod_{C \text{ component } of F}(-1)^{|C|}q^{|V(C)|-1}.
\end{align*}
Next, we break up the sum over $F \subseteq E$ in terms of the component structure of $F$ as follows. We sum over all $F$ that have exactly $k$ connected components with vertex sets $S_1, \ldots, S_k$ and then we sum over all possible choices of $S_1, \ldots, S_k$ and all possible choices of $k$. In fact we can ignore the components that consist of a single vertex (and no edge) since they contribute a factor of $1$ to the product above.
In this way (after exchanging a sum and product) we obtain
\[
q^{|V|}\chi_G(1/q)
=\sum_{k\geq 0} \sum_{\substack{S_1,\ldots,S_k\subseteq V\\ S_i\cap S_j =\emptyset \text{ for } i\neq j \\|S_i|\geq 2}}\prod_{i=1}^k \sum_{\substack{F_i\subseteq E(S_i)\\ (S_i,F_i) \text{ connected}}}(-1)^{|F_i|} q^{|S_i|-1}.
\]
By construction any collection of sets $\{S_1,\ldots,S_k\}$ contributing to this sum forms an independent set of size $k$ in the graph $\Gamma$. The weights are constructed precisely so that the last expression is $Z_\Gamma((\lambda_S))$, as desired.
\end{proof}

\subsubsection{Zero-freeness conditions and their verification}
Here we present a result due to Biascot, F\'ernandez and Procacci~\cite{BFP} that provides useful conditions that guarantee that our multivariate independence polynomial (for graphs of the type $\Gamma$)
does not evaluate to zero. We will then verify these conditions for our situation. Let $G=(V,E)$ and $\Gamma$ be as before.
\begin{theorem}[\cite{BFP}] \label{thm:GK}
For any complex numbers $(\lambda_S)_{S \in V(\Gamma)}$ and any $a>1$, if,
 for each $v\in V$, it holds that
\begin{equation}\label{eq:GK condition}
\sum_{\substack{S\mid v\in S \\ |S| \geq 2} }|\lambda_S| a^{|S|}\leq a-1,
\end{equation}
then $Z_\Gamma((\lambda_S))\neq 0$. 
\end{theorem}


The theorem can be proved along the same lines as the proof of Theorem~\ref{thm:shearer}. See~\cite[Proposition 3.1]{BFP} for a proof along these lines and a discussion of how this condition compares with other similar conditions including the Kot\'ecky-Preis conditions~\cite{KP86} and Dobrushin's conditions~\cite{Dobrushin96}.

To verify the conditions in Theorem~\ref{thm:GK}, we need a bound on the weights $\lambda_S$ given in \eqref{eq:weights}.
Our first step in this direction is to get rid of the `alternating signs' in \eqref{eq:weights}. The lemma below can for example be proved using well-known properties of the Tutte polynomial; see e.g.~\cite{ellis2022handbook} for these properties and see~\cite{Penrose,Sokalzeros} for a direct proof.
\begin{lemma}\label{lem:spanning trees}
Let $H$ be a connected graph and denote by $\tau(H)$ the number of spanning trees in $H$. Then
\[
\Big|\sum_{\substack{F\subseteq E(H)\\ (V(H),F) \text{ connected}}}(-1)^{|F|}\Big|\leq \tau(H).
\]
\end{lemma}

For a graph $G=(V,E)$, a vertex $v\in V$, and a variable $x$ we define the \emph{tree generating function} by
\[
T_{G,v}(x):=\sum_{\substack{T\subseteq E(G)\\ (V(T),T) \text{ is a tree, } v\in V(T)}} x^{|T|}.
\]
We can now bound $\sum_{S\mid v\in S, |S| \geq 2} |\lambda_S|a^{|S|}$ in terms of the tree generating function as follows:
\begin{align}
\sum_{S\mid v\in S, |S| \geq 2} |\lambda_S|a^{|S|} 
&= \sum_{S\mid v\in S, |S| \geq 2} \Bigg| \sum_{\substack{F\subseteq E(S)\\ \text{connected}}} (-1)^{|F|} q^{|S|-1} \Bigg| a^{|S|} \nonumber \\
&\leq \sum_{S\mid v\in S, |S| \geq 2} \Bigg| \sum_{\substack{F\subseteq E(S)\\ \text{connected}}} (-1)^{|F|} \Bigg| |q|^{|S|-1}  a^{|S|} \nonumber \\
&\leq \sum_{S\mid v\in S, |S| \geq 2} \tau(G[S]) |q|^{|S|-1}  a^{|S|} \nonumber \\
&= aT_{G,v}(a|q|) - a. \label{eq:weightbound}
\end{align}

The next lemma shows how to bound the tree generating function. The proof we give is  slightly shorter than the proof given in~\cite{JPS13}, and is new as far as we know.
\begin{lemma}[\cite{HR21}]
Let $G=(V,E)$ be a graph of maximum degree at most $\Delta\geq 1$ and let $v\in V$.
Fix any $\alpha>1$.
Then 
\[
T_{G,v}\left(\tfrac{\ln \alpha}{\alpha \Delta}\right)\leq \alpha.
\]
\end{lemma}
\begin{proof}
The proof is by induction on the number of vertices of $G$. If $|V| = 1$, the statement is clearly true.
Next assume that $|V|\geq 2$.
Given a tree $T$ such that $v\in V(T)$ let $S$ be the set of neighbours of $v$ in $V(T)$.
After removing $v$ from $T$, the tree decomposes into the disjoint union of $|S|$ trees, each containing a unique vertex from $S$.
Therefore, writing $c=\tfrac{\ln \alpha}{\alpha \Delta}$, we have
\[
T_{G,v}(c)\leq \sum_{S\subseteq N_G(v)} c^{|S|}\prod_{s\in S}T_{G-v,s}(c),
\]
which by induction is bounded by
\[
\sum_{S\subseteq N_G(V)} (c\alpha)^{|S|} \leq (1+(\ln \alpha)/\Delta)^\Delta\leq e^{\ln\alpha}=\alpha.
\]
This finishes the proof.
\end{proof}

 
We now combine all our ingredients to finish the proof of Theorem~\ref{thm:JPS}. Fix $\Delta\geq 2$. For $a>1$ to be determined, define $\alpha=\alpha(a)=2-1/a$.
Then if 
\[
|q|\leq \frac{\ln\alpha}{a \alpha \Delta}=\frac{\ln(2-1/a)}{(2a-1)\Delta},
\]
we have $\chi_G(1/q)\neq0$ for any graph of maximum degree at most $\Delta$.
Indeed, for such a value of $q\neq 0$ we have $|aq|\leq \tfrac{\ln \alpha}{\alpha \Delta}$ and therefore by \eqref{eq:weightbound} and the previous lemma
\[
\sum_{S\mid v\in S, |S| \geq 2} |\lambda_S|a^{|S|} \leq a(T_{G,v}(a|q|)-1)\leq a(\alpha-1)=a(1-1/a)=a-1,
\]
and so by Theorem~\ref{thm:GK} and Lemma~\ref{le:ChromIndep} we have $\chi_G(1/q)\neq 0$.
In other words if $|q|\geq \Delta\tfrac{2a-1}{\ln(2-1/a)}$ we have $\chi_G(q)\neq0$.
One can determine
\[
\min_{a>1}\frac{2a-1}{\ln(2-1/a)}<6.91,
\]
where the minimum is attained at $a\cong 1.588$.
This finishes the proof sketch of Theorem~\ref{thm:JPS}.

\subsubsection{Recipe and relation to cluster expansion}
The steps we took to prove Theroem~\ref{thm:JPS} suggest a `recipe' for proving absence of zeros using the polymer approach:
\begin{itemize}
\item[$\bullet$] Express the graph polynomial as an evaluation of the  multivariate independence polynomial of an associated graph.
\item[$\bullet$] Use the conditions from Theorem~\ref{thm:GK} (or other conditions) that guarantee the evaluation is nonzero.
\item[$\bullet$] Verify these conditions using combinatorial arguments.
\end{itemize}

Most combinatorial applications of this `recipe' include  various extensions and variations of the chromatic polynomial. See~\cite{Sokalzeros,FerProc, DongKohreal,JS09,JPS13,Dongjonespol,CsikFren,CDKPR} for some examples in this direction.

From a statistical physics perspective both Theorem~\ref{thm:JPS} and Theorem~\ref{thm:shearer} are statements about so-called high temperature models (in the case of Theorem~\ref{thm:shearer}, high temperature means small values of $\lambda$ for the independence polynomial).
Surprisingly, for some restricted families of graphs, the `recipe' above can also sometimes be used at low temperature (see e.g.~\cite{BorgsImbrie,FVbook} for this in statistical physics). For example, it has been used in combination with the interpolation method to design efficient approximation algorithms to approximate the independence polynomial at large $\lambda$ on certain subgraphs of the integer lattice $\mathbb{Z}^d$~\cite{HPR21}, but also on bipartite expander graphs~\cite{JKP20}. 
In fact,~\cite{JKP20} slightly modified the approach from~\cite{HPR21}. 
The idea is to use conditions like those in Theorem~\ref{thm:GK} to show absolute convergence of the \emph{cluster expansion}, a formal power series of the logarithm of $Z_\Gamma((\lambda_S))$, and to bound the remainder after truncating it at a suitable depth.
This avoids the use of the interpolation method and may occasionally lead to faster algorithms, but other than that is quite similar in spirit. See~\cite{BHHPTpottsall,CanPer,LLLM,helmuth2020finite,carlson2020efficient,MannHelmuth,GGS21,friedrich2020polymer,helmuth2022efficient,holantcomplexplane} for some results inspired by and based on this.

\section{Concluding remarks}
\label{se:conc}
We have shown how absence of zeros allows one to design efficient algorithms to approximately compute evaluations of graph polynomials using Barvinok's interpolation method. A key part of this method is establishing absence of zeros for the graph polynomials in question. 
A few natural questions that remain are: how do other approaches for approximate counting relate to absence of zeros, and what does presence of zeros mean for the possibility to design efficient approximation algorithms.
In this section we will briefly address these two questions pointing the interested reader to the relevant literature.

\subsection{Absence of zeros and other algorithm approaches}
As mentioned in the introduction there are two other (and older) approaches for designing approximation algorithms to compute evaluations of graph polynomials: a Markov chain based sampling approach and the method of correlation decay.
We will not discuss the workings of these approaches here, but we mention how these approaches relate to the interpolation method, or rather, how they relate to absence of zeros.

Recently it was shown that a standard technique for proving decay of correlations can be transformed to prove absence of zeros near the real axis~\cite{contraction,liu2019correlation}.
In the other direction, some results appeared indicating that absence of zeros can be used to establish some form of decay of correlations~\cite{gamarnik2022correlation,regts2021absence}.
Perhaps more surprisingly, in~\cite{anari,vigodastability} it was shown that if a multivariate version of the polynomial has no zeros near the positive real axis, then the associated Glauber dynamics (a local Markov chain often used in approximate counting and sampling) mixes rapidly.
These results indicate that, while absence of complex zeros is vital for the interpolation method, it also plays a key role (albeit in disguise) in these two other approaches for approximate counting.

\subsection{Presence of zeros}
In this section we discuss how presence of zeros is related to hardness of approximation.
We will again specialise the discussion to the independence polynomial and give some references to results on other polynomials at the end of this section.
In what follows we shall see that presence of zeros implies hardness of approximating the independence polynomial.

Let us first state the precise algorithmic problem in question.
Let $\lambda\in \mathbb{Q}[i]$ (the set of complex numbers whose real and imaginary parts are both rational) and let $\Delta\in \mathbb{N}$. Consider the following computational problem.
\begin{itemize}
\item[\emph{Name}] \#Hard-CoreNorm($\lambda,\Delta)$
\item[\emph{Input}] A graph $G$ of maximum degree at most $\Delta$.
\item[\emph{Output}] If $Z_G(\lambda)\neq0$ the algorithm must output a rational number $N$ such that $N/1.001\leq |Z_G(\lambda)|\leq 1.001N$. If $Z_G(\lambda)=0$ the algorithm may output any rational number.
\end{itemize}
It is easy to show that replacing the constant $1.001$ by any other constant $C>1$ does not change the complexity of the problem.\footnote{ An algorithm that solves the problem above in polynomial time can also be used to solve the problem with $1.001$ replaced by $(1.001)^2$ by running the original algorithm on the disjoint union of two copies of the graph.}


The typical notion of hardness one considers for computational counting problems is  \#P-completeness/hardness. We do not introduce the notion formally, but wish to impress only that one does not expect a polynomial-time algorithm for a \#P-complete counting problem (just as one does not expect a polynomial-time algorithm for an NP-complete problem). For example, the problem of exactly counting the number of independent sets of a graph of maximum degree $\Delta$ is known to be \#P-complete~\cite{Roth,Vadhan,DyerGreenhill} for any $\Delta \geq 3$. 


Returning to the problem \#Hard-CoreNorm($\lambda,\Delta$), we define
 the sets
\begin{align*}
\mathcal{P}_\Delta&=\{\lambda\in    \mathbb{Q}[i]\mid   \text{ \#Hard-CoreNorm($\lambda,\Delta$) is \#P-hard}\},
\\
\mathcal{Z}_\Delta&=\{\lambda\in\mathbb{C}\mid Z_G(\lambda)=0 \text{ for some graph }G\in \mathcal{G}_\Delta\}.
\end{align*}
Building on~\cite{BGGS20}, it was shown in~\cite{de2021zeros} that the closure of the set $\mathcal{Z}_\Delta$ is contained in the closure of the set $\mathcal{P}_\Delta$, meaning that arbitrarily close to any zero $\lambda\in \mathcal{Z}_\Delta$ there exists a parameter $\lambda'\in \mathbb{Q}[i]$ such that \#Hard-CoreNorm($\lambda',\Delta)$ is \#P-hard. (A similar result holds if instead of approximating the norm one wishes to approximate the argument of $Z_G(\lambda)$.)
Recall that Theorem~\ref{thm:shearer} gives a zero-free region for the independence polynomial that contains the point $0$. If one can show that the \emph{maximal} zero-free region of the independence polynomial for bounded degree graphs is connected, then this would result in an essentially complete understanding of the hardness of approximating the independence polynomial at complex parameters on bounded degree graphs~\cite{de2021zeros}. 
We remark that quite recently it was shown that in the $\Delta\to \infty$ limit this is true~\cite{bencs2021limit}, but unfortunately this does not give us information for any finite $\Delta$ yet.

For some models/polynomials such as the matching polynomial~\cite{BGGS21} and the ferromagnetic Ising model (as a function of the external field)~\cite{BGPR22} we know that absence/presence of zeros on bounded degree graphs exactly corresponds to easiness/ hardness of approximation. For others such as the Ising model (as a function of edge interaction)~\cite{galanis2021complexity} a partial correspondence has been established. 
This suggests a program of study to understand this connection for more models and also to understand the phenomenon in general.

The interpolation method is clearly a powerful technique for establishing fast approximation algorithms to evaluate graph polynomials at complex parameters, but more than that, it often seems to capture the dichotomy between parameters where approximate computation is easy and hard.




\end{document}